\documentclass[letterpaper, 10 pt, conference]{ieeeconf}  

\IEEEoverridecommandlockouts                              
\usepackage{cite}
\usepackage{amsmath,amssymb,amsfonts}
\usepackage{algorithm} 
\usepackage{algpseudocodex}
\usepackage{graphicx}
\usepackage{textcomp}
\usepackage{xcolor}
\usepackage{url}
\usepackage{dblfloatfix}
\usepackage{mathtools}
\usepackage{pgfplots}
\pgfplotsset{compat=newest}

\usepackage{nccmath}

\let\labelindent\relax
\usepackage{enumerate}
\usepackage[shortlabels]{enumitem}

\usepackage{acronym}
\acrodef{MPC}{Model Predictive Control}
\acrodef{IMMPC}{Internal Model MPC}
\acrodef{RPI}{robust positively invariant}

\usepackage{tikz}
\usetikzlibrary{calc}
\usetikzlibrary {arrows.meta} 

\newtheorem{theorem}{Theorem}

\newtheorem{lemma}[theorem]{Lemma}

\newtheorem{definition}{Definition}

\usepackage{tcolorbox}

\title{\LARGE \bf
Robust IMMPC: An Offset-free MPC for Rejecting Unknown Disturbances
}

\author{Felix Brändle and Frank Allgöwer
\thanks{F. Brändle thanks the International Max Planck Research School for Intelligent Systems (IMPRS-IS) for supporting him.}
}

\begin{document}

\maketitle
\thispagestyle{empty}
\pagestyle{empty}

\begin{abstract}
Output regulation is the problem of finding a control input to asymptotically track reference trajectories and reject disturbances.
This can be addressed by using the internal model principle to embed a model of the disturbance in the controller.
In this work, we present a \acl{MPC} scheme to achieve offset-free control.
To do so, we extend \acl{IMMPC} to general bounded disturbances that must not be generated by the disturbance model.
We show recursive feasibility, constraint satisfaction, and provide convergence conditions for the optimal reachable output.
The proposed controller is validated on a four-tank system.
\end{abstract}

\section{Introduction}
\ac{MPC} is a powerful control method, where the control input is determined by solving an optimal control problem at each time step and applying the first optimal input to the plant \cite{Rawlings1994}.
This allows the user to directly specify control objectives via the cost function and to include constraints, such as actuator limits, to avoid damaging the plant or control system.
However, to guarantee stability and constraint satisfaction, an accurate prediction model is needed \cite{Mayne2000}.
This also requires knowing all disturbances affecting the plant.

One approach to handle unknown disturbances is tube-based \ac{MPC} \cite{Mayne2005, Chisci2001}.
If the disturbance is contained within a bounded set, it is possible to design a pre-stabilizing controller to quantify and bound the effect of the disturbance using \ac{RPI}-sets.
Similarly, stochastic \ac{MPC} employs probabilistic bounds to characterize the effect of random disturbances \cite{schlueter2022}.
By performing a constraint tightening, it is possible to guarantee recursive feasibility, stability, constraint satisfaction, and convergence to a set around the reference \cite{Limon2010}. 
However, for non-vanishing disturbances, the plant does not converge exactly to the desired reference.
To overcome this limitation, offset-free \ac{MPC} can be used.
These methods use the internal model principle to embed a model of the disturbance in the controller \cite{Francis1976}.
This allows the controller to asymptotically reject the effect of the disturbance on the output.
A common approach is to estimate the disturbance using an observer \cite{Pannocchia2003,Pannocchia2015}. 
The estimated disturbance is then used in the prediction model to compute the  desired state and input trajectories, which achieve for offset-free control \cite{Morari2012}.
However, this requires the \ac{MPC} to remain feasible until the disturbance estimate has converged \cite{Maeder2010}.
A different approach uses the so-called velocity or incremental form for prediction \cite{Pannocchia2015a}. 
Instead of optimizing over the input directly, the \ac{MPC} optimizes over the rate of change of the inputs.
The key advantage is that the underlying prediction model in velocity form is independent of constant disturbances.
However, to account for constraints, an estimate of the disturbance is still required \cite{koehler2022}.
Other approaches in velocity form only consider offset-free control with constant references, but do not include constant disturbances in the dynamics \cite{Betti2012}.
Interestingly, as shown in \cite{Pannocchia2015a}, the velocity form embeds a special disturbance observer in its models, correlating the velocity form to disturbance observers.
Despite this equivalence, embedding the observer directly in the \ac{MPC} can be beneficial.
For example, one can apply standard \ac{MPC} techniques to ensure recursive feasibility and constraint satisfaction, without having to consider the interaction between \ac{MPC} and an external observer.
In \cite{braendle2025}, it was shown that for systems without constraints, this leads to the well-known state feedback controller with integrator.
Recent work extended the velocity form to embed a broader class of disturbances \cite{braendle2025}.
This scheme is called \ac{IMMPC} and allows for offset-free control with more general disturbances generated by a known linear system, such as ramps or sinusoids, not only constants.
Furthermore, it introduces additional dynamic filters to extend the state feedback with an integrator structure to a more general formulation by replacing the integrator for example by a PID-architecture.
This provides additional degrees of freedom to improve tracking behavior and disturbance rejection. 

In this work, we extend \ac{IMMPC} to account for disturbances, that are not generated by a linear signal generator.
To this end, we combine \ac{IMMPC} with a tube-based approach to be robust against general bounded disturbances.
We show recursive feasibility, constraint satisfaction, and provide conditions for convergence to the optimal reachable steady state trajectory.
Our approach is also robust against changes in the disturbance generated by the signal generator, such as setpoint changes.
Lastly, we apply the proposed controller to a four-tank system showcasing its capabilities in tracking and rejecting disturbances.

\emph{Notation}
The matrix $I_n$ is the identity matrix of dimension $n$. We omit the index $n$, if is clear from the context.
We denote $\|x\|^2_Q = x^\top Q x$.
We use $x([t_1,t_2])$ with $t_1\in\mathbb{N}$ and $t_2\in\mathbb{N}$ for the stacked vector of $x(t_1)$, $x(t_1-1)$, $\ldots$ $x(t_2)$.
Moreover, we use $P\succ0$ ($P\succeq0)$, if the matrix $P$ is positive (semi-) definite.
Similarly, we use $P\prec0$ ($P\preceq0)$ for negative (semi-) definiteness.
We use $p(z)\{x(t)\}=0$ with $p(z)=\sum_{i=0}^{n_p}p_iz^{-i}$ if $\sum_{i=0}^{n_p}p_i x(t-i)=0$ for all $t\in\mathbb{N}$.
The Minkowski sum and the Pontraygin difference are denoted by $\oplus$ and $\ominus$ and $\sigma\mathbb{X}\coloneq\{\sigma x  \mid x\in\mathbb{X}\}$ with $\sigma\in\mathbb{R}$.

\section{Setup}
In this work, we consider a linear, time-invariant system
\begin{subequations}
	\begin{align}
		  x(t+1) &= Ax(t) + v_x(t) + w_x(t) + Bu(t)\label{eq:Setup:Dynamic}\\
		  e(t) &= Cx(t) + v_e(t) + w_e(t) \label{eq:Setup:Output} 
	\end{align}
\end{subequations}
with $A\in\mathbb{R}^{n \times n}$, $B\in\mathbb{R}^{n \times m}$, $C\in\mathbb{R}^{p \times n}$, state $x(t)\in\mathbb{R}^n$, measurable output $e(t)\in\mathbb{R}^p$, control input $u(t)\in\mathbb{R}^{m}$, and the unknown exogenous disturbances $v_x(t)\in\mathbb{R}^{n}$, $v_e(t)\in\mathbb{R}^{p}$, $w_x(t)\in\mathbb{R}^{n}$, and $w_e(t)\in\mathbb{R}^{p}$.
We assume $(A,B)$ to be controllable and $(A,C)$ to be observable.
Furthermore, the system is subject to state and input constraints $x(t)\in\mathbb{X}\subseteq\mathbb{R}^n$ and  $u(t)\in\mathbb{U}\subseteq\mathbb{R}^m$ for all $t\in\mathbb{N}$ with $\mathbb{X}$ and $\mathbb{U}$ being compact sets.
In this paper, we investigate the output regulation problem.
We aim to design a controller that rejects the effect of unknown disturbances on the output and steers $e(t)$ to zero.
To this end, we split the total disturbance affecting the system into $v_x(t)$, $v_e(t)$, $w_x(t)$ and $w_e(t)$.
The disturbances $v_x(t)$ and $v_e(t)$ are assumed to be generated by a known stable, linear signal generator \cite{Francis1976}
\begin{align}
	\sum_{i=0}^{n_p}p_i v_x(t-i)=0&& \sum_{i=0}^{n_p}p_i v_e(t-i)=0. \label{eq:Setup:DynamicDisturbance}
\end{align}
Alternatively, this can be expressed in terms of a transfer function $p(z)=\sum_{i=0}^{n_p}p_iz^{-i}$ with $p_0=1$ and $p_{n_p}\neq 0$, such that $p(z)\{v_x(t)\}=0$ and $p(z)\{v_e(t)\}=0$. 
For example $p(z)=(1-z^{-1})$ characterizes constant signals and $p(z)=(1-2\cos(\omega_0)z^{-1}+z^{-2})\{v(t)\}=0$ characterizes sinusoids with frequency $\omega_0$.
Furthermore, we consider $w_x(t)$ and $w_e(t)$ with
\begin{align}
	w_x(t)\in\mathbb{W}_x&&w_e(t)\in\mathbb{W}_e
\end{align}
for all $t\in\mathbb{N}$, and $\mathbb{W}_x\subseteq\mathbb{R}^n$ and $\mathbb{W}_e\subseteq\mathbb{R}^p$ being compact sets, which contain the origin as an interior point \cite{Mayne2005}. 
Hence, we assume $w_x(t)$ and $w_e(t)$ to be bounded, but do not impose any condition on their dynamics.
In order to limit the effect of bounded disturbances, we introduce \ac{RPI}-sets.
\begin{definition}
	A set $\mathbb{S}\subseteq\mathbb{R}^n$ is called a \acl{RPI} set for $A\in\mathbb{R}^{n \times n}$, $B\in\mathbb{R}^{n \times n}$, and $\mathbb{W}\subseteq\mathbb{R}^n$, if $A x + Bw \in \mathbb{S}$ for all $x\in\mathbb{S}$ and all $w\in\mathbb{W}$.
\end{definition}
By differentiating between the two classes of disturbances, we can split the total disturbance affecting the system into a potentially large disturbance with a known signal generator and an arbitrarily generated, but bounded disturbance.

Lastly, we assume that the unconstrained output regulation problem is well-posed. This means that there exists a unique state and input trajectory satisfying
\begin{equation}
\begin{aligned}
	x(t\!+\!1) &= A x(t) + v_x(t) + Bu(t) \\
	0 &= C x(t) + v_e(t) \label{ass:Setup:Sylvester}
\end{aligned} 
\end{equation}
for all $t\in\mathbb{N}$ with $p(z)\{x(t)\}=0$, $p(z)\{u(t)\}=0$, $p(z)\{v_x(t)\}=0$ and $p(z)\{v_e(t)\}=0$.
In this work, we combine tube-based \ac{MPC} \cite{Mayne2005} with \ac{IMMPC} \cite{braendle2025} to steer $e(t)$ to zero, despite the system being affected by unknown disturbances.

\section{MPC}
In this section, we design a robust \ac{MPC} for output regulation.
Section\,\ref{sec:MPC:Model} introduces an equivalent representation of the system without requiring $v_x(t)$ and $v_e(t)$ explicitly. 
Section\,\ref{sec:MPC:Controller} and \ref{sec:MPC:RPI} describe how to compute \ac{RPI}-sets for the disturbances. 
Section\,\ref{sec:MPC:MPC} states the final \ac{MPC}.

\subsection{Model} \label{sec:MPC:Model}
In this section, we derive a different representation of \eqref{eq:Setup:Dynamic}, \eqref{eq:Setup:Output}, and \eqref{eq:Setup:DynamicDisturbance}, using only past state, input, and output data. 
\begin{lemma}[\!\!\cite{braendle2025}, Theorem 1] \label{thm:MPC:Model}
	Suppose $x([0,-n_p])$, $u([-1,-n_p])$ and $e([0,1-n_p])$ are generated according to \eqref{eq:Setup:Dynamic}, \eqref{eq:Setup:Output}, and \eqref{eq:Setup:DynamicDisturbance}, then 
	\begin{subequations}
		\begin{align}
			\!\!\!\!\!\!\Delta x(t\!+\!1) \!&=\! A \Delta x(t) + B \Delta u(t) + \!\Delta w_x(t) \label{eq:MPC:Theo:ExtendedDeltaXDyn}\\
			\!\!\!\!\!\!e(t\!+\!1) \!&=\! C \Delta x(t\!+\!1)\!-\!\!\sum_{i=0}^{n_p-1}\! p_{i+1} e(t\!-\!i)\! + \!\Delta w_e(t\!+\!1) \label{eq:MPC:Theo:ExtendedEDyn} \\
			\!\!\!\!\!\!x(t\!+\!1) \!&=\! \Delta x(t\!+\!1)-\!\sum_{i=0}^{n_p-1}\! p_{i+1} x(t-i) \label{eq:MPC:Theo:ExtendedXDyn}\\ 
			\!\!\!\!\!\!u(t) \!&=\! -\!\sum_{i=1}^{n_p}\! p_i u(t-i)\! + \!\Delta u(t), \label{eq:MPC:Theo:ExtendedUDyn}
		\end{align}
	\end{subequations}
	hold for all $t\in\mathbb{N}^0$ with $\Delta x(t) = \sum_{i=0}^{n_p} p_i x(t-i)$, $\Delta w_x(t) \!= \!\sum_{i=0}^{n_p} p_i w_x(t-i)$, and $\Delta w_e(t) \!=\! \sum_{i=0}^{n_p} p_i w_e(t-i)$.
\end{lemma}\vspace{2pt}
\begin{proof}
	To show \eqref{eq:MPC:Theo:ExtendedDeltaXDyn}, we consider $\sum_{i=0}^{n_p} p_i x(t\!+\!1-i)$ and insert \eqref{eq:Setup:Dynamic} with $\Delta u(t)=\sum_{i=0}^{n_p} p_i u(t\!-i)$.
	Due to \eqref{eq:Setup:DynamicDisturbance}, \eqref{eq:MPC:Theo:ExtendedDeltaXDyn} follows directly. 
	\eqref{eq:MPC:Theo:ExtendedEDyn} holds by the same arguments.
	\eqref{eq:MPC:Theo:ExtendedXDyn} and \eqref{eq:MPC:Theo:ExtendedUDyn} follow by inverting the difference equations.
\end{proof}
Lemma\;\ref{thm:MPC:Model} is an equivalent description of \eqref{eq:Setup:Dynamic}, \eqref{eq:Setup:Output}, and \eqref{eq:Setup:DynamicDisturbance} that does not explicitly require $v_x(t)$ and $v_e(t)$.
This will serve as the prediction model for the \ac{MPC}.
However, Lemma\;\ref{thm:MPC:Model} needs past state, input, and output measurements.
Hence, any change in $v_x(t)$ and $v_e(t)$ not described by \eqref{eq:Setup:DynamicDisturbance}, requires a new initialization phase of $n_p$ steps.
The dynamics are still linear and can be described by a state space system with extended state $\xi(t)\in\mathbb{R}^{n_\xi}$, control input $\Delta u(t)$:
\begin{align*}
	\xi(t\!+\!1) = A_\xi \xi(t) + B_\xi \Delta u(t) + B_{1}w(t) + B_{2}\mathbf{w}(t),
\end{align*}\vspace{-20pt}
\begin{align*}
	\xi(t)\!\coloneq\!\begin{bmatrix}
		\Delta x(t) \\
		e([t,t-n_p+1]) \\
		x([t,t-n_p+1]) \\
		u([t-1,t-n_p])
	\end{bmatrix},
\end{align*}
with $w(t)=[w_x^\top(t), w_e^\top(t+1)]^\top\in\mathbb{W}\coloneq \mathbb{W}_x \times \mathbb{W}_e$ and $v(t)=[v_x^\top(t), v_e^\top(t+1)]^\top$.
Furthermore, we define the vector of past disturbances $\mathbf{w}(t)=w([t-1,t-n_p])$ and $\mathbf{v}(t)=v([t-1,t-n_p])$ with $\mathbf{v}(t+1)=S\mathbf{v}(t)$ and $S\in\mathbb{R}^{n_p(n+p)\times n_p(n+p)}$ as state space realization of $p(z)$.
In this paper, we refer to $\Delta x(t)$, $x(t)$, $e(t)$ and $u(t)$ directly or use $\Delta x(t) = C_{\Delta x}\xi(t)$, $x(t) = C_{x}\xi(t)$, $e(t) = C_{e}\xi(t)$, $u(t-1) = C_{u}\xi(t)$ and $e([t,t-n_p+1]) = M_e\xi(t)$.
Note that the extended state $\xi(t)$ is not stabilizable as it embeds the non-stabilizable disturbance state $\mathbf{v}(t)$
\begin{equation}\label{eq:MPC:TiDefinition}
	\begin{split}
	T_i \xi(t)&\coloneq \begin{bmatrix}
		x(t+1-i)-Ax(t-i)-Bu(t-i)\\
		e(t+1-i)-Cx(t+1-i)
	\end{bmatrix}\\
	&\;= w(t-i)+v(t-i) 
\end{split}
\end{equation}
with $i=1,\ldots,n_p$, $T_i\in\mathbb{R}^{(n+p)\times n_\xi}$ and $p_{n_p}x(t-n_p) = \Delta x(t) - \sum_{i=0}^{n_p-1}p_ix(t-i)$.
Combining \eqref{eq:MPC:TiDefinition} for every $i$ yields $\mathbf{v}(t)+\mathbf{w}(t)=T_v\xi(t)$ with $T_v = [T_1^\top,\ldots,T_{n_p}^\top]^\top$.

\subsection{Stabilizing Controller}  \label{sec:MPC:Controller}
In order to bound the effect of $w(t)$, we must first design a stabilizing controller \cite{Mayne2005}.
As the full extended system is not stabilizable, we design a controller to stabilize only \eqref{eq:MPC:Theo:ExtendedDeltaXDyn} and \eqref{eq:MPC:Theo:ExtendedEDyn} to achieve output regulation.
We consider \eqref{eq:MPC:Theo:ExtendedXDyn} and \eqref{eq:MPC:Theo:ExtendedUDyn} only for constraint satisfaction.
The subsystem \eqref{eq:MPC:Theo:ExtendedDeltaXDyn} and \eqref{eq:MPC:Theo:ExtendedEDyn} is stabilizable, as $(A,B)$ is controllable and \eqref{ass:Setup:Sylvester} ensure the existence of a control input to reach $e(t)= 0$ with $\Delta x(t)=0$ and $\Delta u(t)=0$.
Hence, we can design a state feedback controller 
\begin{equation}
	\Delta u(t) = K_x \Delta x(t) + K_e e([t,t\!+\!1-n_p]) \label{eq:MPC:StabilizingController}
\end{equation}
to stabilize \eqref{eq:MPC:Theo:ExtendedDeltaXDyn} and \eqref{eq:MPC:Theo:ExtendedEDyn} with $K_x\in\mathbb{R}^{m \times n}$ and $K_e\in\mathbb{R}^{m \times p n_p}$. 
We can recover $u(t)$ from 
\begin{align}
	u(t) = \Delta u(t) - \sum_{i=1}^{n_p} p_i u(t-i) = K\xi(t) \label{eq:MPC:StateFeedbackControllerDifferenceEq}
\end{align}
with $K\in\mathbb{R}^{m \times n_\xi}$.
Alternatively, we can express the controller as a state space model
\begin{align}
	x_{\mathrm{c}}(t+1) &= A_{\mathrm{c}}x_{\mathrm{c}}(t) +  B_{\mathrm{c}}e(t) \label{eq:MPC:ControllerDyn}\\
	u(t)&=C_{\mathrm{c}}x_{\mathrm{c}}(t) + D_{\mathrm{c}}e(t) + K_x x(t) \label{eq:MPC:ControllerDynOutput}
\end{align}
with $x_{\mathrm{c}}(t)=T_{\mathrm{c}}\xi(t)\in\mathbb{R}^{pn_p}$. 
For $p(z)=1-z^{-1}$, it holds that $A_{\mathrm{c}}=I$, $B_{\mathrm{c}}=K_e$, $C_{\mathrm{c}}=I$, $D_{\mathrm{c}}=K_e$ with $x_{\mathrm{c}}(t)=u(t-1)-K_x (x(t)-\Delta x(t))$. 
Since $A_{\mathrm{c}}=I$, this connects to the classical state feedback controller with integrator structure.
As shown in \cite{braendle2025}, by selecting a different filter than $p(z)$, we can further tune the controller to get a PID-controller acting on $e(t)$ instead of only the integrator.
In the remaining paper, we assume access to such a stabilizing controller gain $K$.

\subsection{RPI-set}  \label{sec:MPC:RPI}
Following the approach in \cite{Mayne2005}, we now derive an \ac{RPI}-set to account for $w_x(t)$ and $w_e(t)$.
To do so, we introduce a nominal model with $\Delta w_x(t)=0$ and $\Delta w_e(t)=0$, such that $B_1w(t)+B_2\mathbf{w}(t)=0$.
In contrast to \cite{Mayne2005}, this does not imply that the nominal system evolves according to $w_x(t)=0$ and $w_e(t)=0$. Instead, it evolves according to a nominal disturbance $w_x(t)=-\sum_{i=1}^{n_p}p_i w_x(t-i)$ and $w_e(t)=-\sum_{i=1}^{n_p}p_i w_e(t-i)$.
This stems from the fact, that $\mathbf{v}(t)+\mathbf{w}(t)=T_v\xi(t)$ can not distinguish between $\mathbf{w}(t)$ and $\mathbf{v}(t)$.
We denote the nominal system by $``\mathrm{n}``$, e.g., $\xi_\mathrm{n}(t)$.
By performing a state space transformation $T\xi_\mathrm{n}(t)$ with $T=[T_0^\top,T_v^\top]^\top$ and $T_0=[C_x^\top,T_\mathrm{c}^\top]^\top$, this results in
\begin{align}
		x_\mathrm{n}(t+1) &= A x_\mathrm{n}(t) + Bu_\mathrm{n}(t) + v_{x,\mathrm{n}}(t)\\
		x_\mathrm{c,n}(t+1) &= B_{\mathrm{c}}Cx_\mathrm{n}(t)+ A_{\mathrm{c}}x_\mathrm{c,n}(t) + B_{\mathrm{c}}v_{e,\mathrm{n}}(t)\\
		e_{\mathrm{n}}(t) &= Cx_\mathrm{n}(t) + v_{e,\mathrm{n}}(t)
\end{align}
with $p(z)\{v_{\mathrm{n},x}(t)\}=0$, $p(z)\{v_{\mathrm{n},e}(t)\}=0$ and initialization $\mathbf{v}_{\mathrm{n}}(0)=T_v \xi_\mathrm{n}(0)$.
Next, we apply the control input
\begin{align}
	u(t)&=u_\mathrm{n}(t) + K(\xi(t)-\xi_{\mathrm{n}}(t))\\
	&=u_\mathrm{n}(t)+C_\mathrm{c}\hat{x}_\mathrm{n}(t)+ D_\mathrm{c}\hat{e}(t) + K_x\hat{x}(t).
\end{align}
The index $``\,\,\hat{ }\,``$ denotes the difference, e.g., $\hat{x}(t)=x(t)-x_\mathrm{n}(t)$.
Then the closed loop behaves according to
\begin{align*}
	\begin{bmatrix}
		\hat{x}(t+1)\\\hat{x}_{\mathrm{c}}(t+1)
	\end{bmatrix} &= 
	\begin{bmatrix}
		A+B(K_x + D_\mathrm{c}C) & B C_\mathrm{c} \\
		B_\mathrm{c} C & A_\mathrm{c}
	\end{bmatrix}
	\begin{bmatrix}
		\hat{x}(t)\\ \hat{x}_{\mathrm{c}}(t)
	\end{bmatrix}\\
	&+\begin{bmatrix}
		I & B D_\mathrm{c}\\ 0 & B_\mathrm{c}
	\end{bmatrix}
	\begin{bmatrix}
		\hat{v}_x(t) + w_x(t) \\ \hat{v}_e(t) + w_e(t)
	\end{bmatrix}
\end{align*}
with $p(z)\{\hat{v}_{\mathrm{n},x}(t)\}=0$, $p(z)\{\hat{v}_{\mathrm{n},e}(t)\}=0$ and initialization $\mathbf{\hat{v}}(0)=\mathbf{v}(0)-T_v \xi_\mathrm{n}(0)$. 
Using the controller from Section\,\ref{sec:MPC:Controller}, it is possible to asymptotically stabilize $\hat{x}(t)$ and $\hat{x}_\mathrm{c}(t)$.
However, $\hat{v}(t)$ is not stabilizable as $p(z)$ is not asymptotically stable, such that we can not compute an \ac{RPI}-set for $\xi(t)$.
To address this, we determine an \ac{RPI}-set only for $\hat{x}(t)$ and $\hat{x}_\mathrm{c}(t)$, while enforcing an initialization for $\mathbf{\hat{v}}(0)$ such $\hat{v}_x(t)$ and $\hat{v}_e(t)$ are contained in a bounded set.
To do so, note that $T_v \xi(t)-\mathbf{v}(t)=\mathbf{w}(t)\in\mathbb{W}^{n_p}$.
Hence, if $T_v \xi_\mathrm{n}(t) = T_v \xi(t)$, then $\mathbf{\hat{v}}(0)\in\mathbb{W}^{n_p}$.
Using that $p(z)$ is stable, we conclude that if $\mathbf{\hat{v}}(0)$ is bounded, so are $\hat{v}_x(t)$ and $\hat{v}_e(t)$.
To formalize this, we construct a convex set $\mathbb{V}\subseteq\mathbb{R}^{n_p(n+p)}$, which is invariant with respect to $p(z)$, i.e., if $\hat{\mathbf{v}}(0)\in\mathbb{V}$ , then $S^t\mathbf{\hat{v}}(0)\in\mathbb{V}$ for all $t\in\mathbb{N}$, where $S$ is a state space realization of $p(z)$, see Section\,\ref{sec:MPC:Model}.
Furthermore, $\mathbb{V}$ must satisfy
\begin{align*}
	\big\{\mathbf{v}\!\mid\! \exists \hat{t}\in\mathbb{N}^0, \,\mathbf{v}(0)\in\mathbb{W}^{n_p}\!:\! \mathbf{v}(t\!+\!1)\!=\!S\mathbf{v}(t),\, \mathbf{v}(\hat{t})=\mathbf{v} \big\}\!\subseteq\!\mathbb{V}.
\end{align*}
This ensures that $\mathbb{V}$ contains all possible disturbances $\mathbf{\hat{v}}(t)$ for $\mathbf{\hat{v}}(0)\in\mathbb{W}^{n_p}$.
If $p(z)=1-z^{-1}$ and $\mathbb{W}$ is convex, a possible choice is $\mathbb{V}=\mathbb{W}$, or in general, a level set of the Lyapunov function of $\mathbf{v}(t\!+\!1)\!=\!S\mathbf{v}(t)$.
Moreover, due to convexity, if we have a  vector $\mathbf{\bar{v}}(t)\in \mathbf{v}(t)\oplus\mathbb{V}$, then
\begin{align}
	\tau \mathbf{\bar{v}}(t) + (1-\tau)T_v\xi(t)\in \mathbf{v}(t)\oplus\mathbb{V} \label{eq:MPC:VConvexInterpolation}
\end{align}
for all $\tau\in[0,1]$.
Hence, instead of enforcing  $T_{v}\xi_\mathrm{n}(t)=T_{v}\xi(t)$, we can use $T_{v}\xi_\mathrm{n}(t) = \tau \mathbf{\bar{v}}(t) + (1-\tau)T_v\xi(t)$, to ensure $T_{v}\xi_\mathrm{n}(t)\in\mathbf{v}(t)\mathbb{V}$.
This provides an additional degree of freedom, which will be used in the following section.
Next, we denote $\mathbb{V}_1\coloneq [I_{n+p},0,\ldots,0]\mathbb{V}$ as the set of all possible $\hat{v}_x(t)$ and $\hat{v}_e(t)$.
Finally, we can compute an \ac{RPI}-set $\mathbb{S}$ for $\hat{x}(t)$ and $\hat{x}_{\mathrm{c}}(t)$ with respect to the disturbance set $\mathbb{V}_1\oplus\mathbb{W}$.
That means if $[\hat{x}(t)^\top,\hat{x}_{\mathrm{c}}(t)^\top]^\top\in\mathbb{S}$, then $x(t)-x_\mathrm{n}(t)\in\mathbb{S}_x\coloneq [I,0]\mathbb{S}$ and $u(t)-u_\mathrm{n}(t)\in\mathbb{S}_u\coloneq [K_x + D_\mathrm{c}C,C_\mathrm{c}]\mathbb{S} \oplus D_\mathrm{c}\mathbb{W}_e \oplus D_\mathrm{c}[0, I]\mathbb{V}_1$.
This forms the basis for robust constraint satisfaction in the \ac{MPC}.

\subsection{Robust MPC} \label{sec:MPC:MPC}
In this section, we state the robust \ac{MPC} for output regulation.
The goal is to steer $e(t)$ to the origin.
However, due to $v_x(t)$ and $v_e(t)$, this does not imply that $x(t)$ and $u(t)$ will also be zero. 
To account for this, we introduce artificial references
\begin{align*}
	\mathcal{Z}_{\mathrm{a}} \coloneq \big\{\xi\in\mathbb{R}^{n_\xi}\mid &  \xi(t\!+\!1) = A_\xi \xi(t), \:\:  \xi(0) = \xi, \\
	&C_u \xi(t) \in \mathbb{U}\ominus \sigma\mathbb{S}_u,\:\: C_x \xi(t) \in \mathbb{X}\ominus \sigma\mathbb{S}_x, \\
	& C_{\Delta x}\xi(t)=0, \:\: \forall t\in\mathbb{N}^0\big\}.
\end{align*}
The variable $\sigma>1$ is a design parameter to enforce strict satisfaction of the constraints, as in \cite{Limon2008,Limon2010}.
The set $\mathcal{Z}_{\mathrm{a}}$ describes all possible initial conditions such that $p(z)\{x_{\mathrm{a}}(t)\}=0$, $p(z)\{u_{\mathrm{a}}(t)\}=0$ and $p(z)\{e_{\mathrm{a}}(t)\}=0$, while $x_{\mathrm{a}}(t)$ and $u_{\mathrm{a}}(t)$ lie within the constraint sets.
How to reduce constraint satisfaction from all $t\in\mathbb{N}^0$ to a finite number of verifiable inequalities is described in \cite{Krupa2024} for polytopic sets.
The set $\mathcal{Z}_{\mathrm{a}}$ does not enforce $e_{\mathrm{a}}(t)=0$ to handle the case when $e(t)=0$ is not reachable within the constraints.
Instead, we want to achieve convergence to the optimal reachable reference within the constraints.
To do so, we solve the following optimization problem
\begin{fleqn}[\parindent]
\begin{equation*}
	\begin{aligned}
		V^*(\xi) = \min_{\xi^d(0),\,\Delta u(t)} &&V(M_e\xi^d(0))
	\end{aligned}
\end{equation*}\vspace{-1.4em}
\begin{equation*}
	\begin{aligned}
		&\text{s.t. } &&\xi(t\!+\!1) = A_\xi \xi(t) + B_\xi \Delta u(t) &&\xi(0)=\xi& \\
		&&&\xi^d (t\!+\!1) = A_\xi\xi^d(t) &&\xi^d(0)\in\mathcal{Z}_{\mathrm{a}}&\\
		&&&\lim\limits_{t\to\infty}\xi(t) - \xi^d(t) = 0
	\end{aligned}
\end{equation*}
\end{fleqn}
with $V(e) = \|e\|^2_P$.
As the extended state $\xi(t)$ is not controllable, we also optimize over the input $\Delta u(t)$ to ensure that $\xi^\mathrm{d}(t)$ lies within a reachable manifold.
We denote the optimal initial condition for the reference $\xi^\mathrm{d}(\xi)$ as a function of the extended state.
The matrix $P\succ 0$ is a design parameter. 
In the following, we choose $P$ such that $V(e([t+1,t+2-n_p]))-V(e([t,t-n_p+1])) \leq 0$ for any trajectory with $p(z)\{e(t)\}=0$.
A possible choice is a Lyapunov function of the corresponding stable state space realization.
Next, we introduce the quadratic stage cost
\begin{equation*}
	l(\xi,\Delta u) = \|C_{\Delta x}\xi\|^2_{Q_{\Delta x}} + \|C_{e}\xi\|^2_{Q_e} +\|\Delta u\|^2_{R}.
\end{equation*}
with $Q_{\Delta x}\in\mathbb{R}^{n\times n}$ positive semi-definite, and $Q_{e}\in\mathbb{R}^{p n_p\times p n_p}$ and $R\in\mathbb{R}^{m \times m}$ positive definite to steer $\Delta x(t)$, $\Delta u(t)$, and the output to zero.

Now, we consider the initialization of the nominal state. While enforcing $T_v\xi_\mathrm{n}(t)=T_v\xi(t)$ ensures that $\mathbf{\hat{v}}_\mathrm{n}(t)\in\mathbb{V}$, it does not guarantee that the optimization problem remains feasible, as the embedded disturbance $T_v\xi(t)$ changes depending on the actual $w_x(t)$ and $w_e(t)$.
Inspired by \cite{schlueter2022}, we use convex interpolation from \eqref{eq:MPC:VConvexInterpolation} with $\mathbf{\bar{v}}(t)$ being a feasible fallback solution and $\tau\in[0,1]$ as interpolation variable.
We introduce the cost
\begin{align*}
	l_\tau(\tau) = \lambda \tau^2
\end{align*}
with $\lambda \geq 0$ to penalize deviations from the current extended state $T_v\xi(t)$.
Using all of this, we can now state the \ac{MPC}
\begin{subequations}
	\begin{align} \label{eq:MPC:MPCCost}
		J^*(\xi,\mathbf{\bar{v}}) &= \min_{\Delta u_k,\, \xi_0,\, \xi_{\mathrm{a},0},\tau}  \sum_{k=0}^{N-1} l(\xi_k-\xi_{\mathrm{a},k},\Delta u_k) \nonumber\\
		&+l_\tau(\tau)+ V(M_e \xi_{\mathrm{a},0}) 
	\end{align} \vspace{-20pt}
	\begin{align}
		\text{s.t. }  
			&\xi_{k+1} = A_\xi\xi_{k} + B_\xi \Delta u_{k} &&  k=0,\ldots,N-1\label{eq:MPC:MPCDyn}\\
			&\xi_{\mathrm{a},k+1} = A_\xi\xi_{\mathrm{a},k} && k = 0,\ldots, N-1\\
			&C_x\xi_k \!\in\! \mathbb{X}\!\ominus\!\mathbb{S}_x,\, u_k \!\in\!\mathbb{U}\!\ominus\! \mathbb{S}_u && k =0,\ldots,N \label{eq:MPC:MPCConstraints}\\
			&T_v\xi_0 = \tau\mathbf{\bar{v}} + (1-\tau) T_v\xi && \tau\in[0,1]  \label{eq:MPC:MPCInitDistState} \\
			&T_0(\xi-\xi_0)\in\mathbb{S} \label{eq:MPC:MPCInitContrState}\\
			&\xi_N = \xi_{\mathrm{a},N} \quad \xi_{\mathrm{a},0} \in\mathcal{Z}_{\mathrm{a}}\label{eq:MPC:MPCXTerminal}
	\end{align}
\end{subequations}
with $u^*_0=\Delta u_0^* - \sum_{i=1}^{n_p}p_{i}u_{-i}^*$, $\xi^*_0$ and $\xi^*_1$ being the corresponding optimizers with $\xi_{k}$ and $\Delta u_{k}$ taking the role of the nominal system from Section\,\ref{sec:MPC:RPI}.
The control input follows by solving the optimization problem $J^*(\xi(t),\mathbf{\bar{v}}(t))$ and evaluating
\begin{equation}\label{eq:MPC:ControlLaw}
	\begin{split}
		u(t) &= u^*_0(t) + K(\xi(t) - \xi^*_0(t)) \\
		\mathbf{\bar{v}}(t\!+\!1) &= T_v \xi^*_1(t).
	\end{split}
\end{equation}
This now has the classical structure of a tube-based \ac{MPC}-scheme with artificial references \cite{Limon2010}. 
Now, we show recursive feasibility, constraint satisfaction, and provide convergence conditions.
\begin{theorem} \label{thm:MPC:MainResult}
	Suppose $\mathbf{\bar{v}}(0)\in \mathbf{v}(0)\oplus\mathbb{V}$, $N>n_p(n+p+m)+n$, and $J^*(\xi(0),\mathbf{\bar{v}}(0))$ is feasible with $\xi(0)$ being generated according to \eqref{eq:Setup:Dynamic}-\eqref{eq:Setup:DynamicDisturbance}.
	If the control law \eqref{eq:MPC:ControlLaw} is applied to the system, then the following holds:
	\begin{enumerate}[i)]
		\item $J^*(\xi(t),\mathbf{\bar{v}}(t))$ is feasible for all $t\in\mathbb{N}$. \label{MPC:MPC:MainThmFeasibility}
		\item $x(t)\in\mathbb{X}$ and $u(t)\in\mathbb{U}$ for all $t\in\mathbb{N}$. \label{MPC:MPC:MainThmConstraints}
		\item If $\lambda=0$, then $\lim_{t\to\infty} \xi_0^*(t)-\xi^\mathrm{d}(\xi_0^*(t))=0$ \label{MPC:MPC:MainThmConvergenceLam0}
		\item If $\lambda>0$ and $p(z)\{w(t)\}=0$, then $\lim_{t\to\infty} \xi_0^*(t)- \xi^\mathrm{d}(\xi_0^*(t))=0$ and $\lim_{t\to\infty} e(t)- C_e\xi^\mathrm{d}(\xi_0^*(t))=0$. \label{MPC:MPC:MainThmConvergence}
		\item Suppose $w(t)\in\sigma_1\mathbb{W}$ for all $t\in\mathbb{N}^0$ with $\sigma_1\in[0,1)$. Further, suppose $\mathbf{\bar{v}}(t_0)\in \mathbf{v}(t_0)\oplus \sigma_2\mathbb{V}$ at time $t_0$ with $\sigma_2\in[0,1)$, then there exists an $\epsilon > 0$ such that for all $\Delta v\in\mathbb{R}^{n+p}$ with $\|\Delta v \|< \epsilon$ and $v(t_0)=\Delta v - \sum_{i=1}^{n_p}p_i v(t_0-i)$, properties \ref{MPC:MPC:MainThmFeasibility} - \ref{MPC:MPC:MainThmConvergence} still hold. \label{MPC:MPC:NoArtRobustness}
	\end{enumerate}
\end{theorem}
\begin{proof}
	 Since $\xi(0)$ is generated according to \eqref{eq:Setup:Dynamic}-\eqref{eq:Setup:DynamicDisturbance} we use Lemma\,\ref{thm:MPC:Model} for prediction.
	
	\ref{MPC:MPC:MainThmFeasibility} This follows from standard arguments by using the time-shifted optimal solution of the previous time step and extending it by the optimal artificial reference
	\begin{align*}
		\xi^\circ(t+1) &= [\xi_1^*(t),\ldots,\xi_N^*(t),A_\xi \xi_N^*(t)] \\
		u^\circ(t+1) &= [\Delta u_1^*(t),\ldots,\Delta u_{N-1}^*(t),0]
	\end{align*}
	with $\xi_{\mathrm{a},0}^\circ(t+1) = A_\xi\xi_{\mathrm{a},0}^*(t)$.
	Using $\tau^\circ =1$ allows us to use nominal disturbance state from the previous time step.
	By construction of the \ac{RPI}-sets it also holds $T_0(\xi(t+1)-\xi_1^*(t))\in\mathbb{S}$, such that the candidate solution $^\circ$ is feasible.
	
	\ref{MPC:MPC:MainThmConstraints} follows by the construction of the \ac{RPI}-sets and the constraint tightening in \eqref{eq:MPC:MPCConstraints}. $T_v \xi_0^*(t) \in \mathbf{v}(t) \oplus \mathbb{V}$ follows from invariance of $\mathbb{V}$ with respect to $p(z)$ and convexity. Furthermore, as $T_0 (\xi(t)-\xi_0^*(t))\in\mathbb{S}$, such that it must hold $x(t)-C_x\xi_0^*(t)\in\mathbb{S}_x$ and $u(t)-u_0^*(t)\in\mathbb{S}_u$.
	
	\ref{MPC:MPC:MainThmConvergenceLam0} We consider the previous candidate solution $\circ$. 
	Using standard arguments from \cite{Mayne2000} and $\lambda=0$, it follows
	\begin{align*}
		&J^*(\xi(t+1),\mathbf{\bar{v}}(t+1)) - J^*(\xi(t),\mathbf{\bar{v}}(t))\\ &\leq  
		 -l(\xi_0^*(t)-\xi_{\mathrm{a},0}^*(t),\Delta u_0^*(t)). \label{eq:MPC:Proof:LyapDecrease}
	\end{align*}
	As the stage cost is non-negative, we can conclude $\lim\limits_{t\to\infty} l(\xi_0^*(t)-\xi_{\mathrm{a},0}^*(t),\Delta u_0^*(t)) = 0$.
	Furthermore, as $Q_e\succ 0$, $R\succ 0$ and $(A,C)$ observable with $N>n_p(n+p+m)+n$, it must hold 
	\begin{align*}
		\lim\limits_{t\to\infty} \xi_0^*(t) - \xi_{\mathrm{a,0}}^*(t) = 0 &&
		\lim\limits_{t\to\infty}\Delta u_0^*(t) = 0
	\end{align*}
	Next, we show $\lim_{t\to\infty} \xi_0^*(t)-\xi^\mathrm{d}(\xi_0^*(t))=0$ by a proof of contradiction.
	The proof follows standard arguments for the artificial references \cite{Limon2008,Limon2010,braendle2025}.
	For space reasons, we shorten some steps.
	First, we consider the converged state with $\xi_0^*(t) = \xi_{\mathrm{a,0}}^*(t)$, but $V(M_e\xi_{\mathrm{a,0}}^*(t))> V(M_e\xi^\mathrm{d}(\xi^*_{\mathrm{a},0}(t))$. 
	By choosing $\Delta u_k^*(t) = 0$, we arrive at the optimal trajectory with vanishing stage cost and $J^*(\xi(t),\mathbf{\bar{v}}(t)) = V(M_e \xi_{\mathrm{a,0}}^*(t))$.
	Next, we show that if $V(M_e\xi_{\mathrm{a,0}}^*(t))> V(M_e\xi^\mathrm{d}(\xi^*_{\mathrm{a},0}(t))$, we can construct a different artificial reference, which leads to a lower total cost, contradicting optimality.
	We define the following candidate solution for the artificial reference
	\begin{align*}
		\bar{\xi}_{\mathrm{a},0} = \gamma \xi^*_{\mathrm{a},0}(t) + (1-\gamma)\xi^\mathrm{d}(\xi^*_{\mathrm{a},0}(t))
	\end{align*}
	with $\gamma\in[0,1)$.	
	As $\mathcal{Z}_\mathrm{a}$ is constructed to lie strictly within the constraints, $(A,B)$ being controllable, and $N>n(n_p+1)+n_pm + n_pp$, we can find sufficiently a large $\bar{\gamma}\in[0,1)$ such that for any $\gamma\in(\bar{\gamma},1)$, $\bar{\xi}_{\mathrm{a},0}$ describes an admissible artificial reference with extended state $\bar{\xi}_k$ and input trajectory $\Delta\hat{u}_k$ to reach  $\bar{\xi}_{\mathrm{a},N}$ without leaving $\mathbb{X}\ominus\mathbb{S}_x$ and $\mathbb{U}\ominus\mathbb{S}_u$. Further, it holds
	\begin{align*}
		\sum_{k=0}^{N-1}\!\! l(\bar{\xi}_k(t)\!-\!\bar{\xi}_{\mathrm{a},k}(t), \Delta \bar{u}_k) \!\leq\! \bar{\beta}\|\bar{\xi}_{\mathrm{a},0}(t)\!-\!\xi^*_{\mathrm{a},0}(t)\|^2
	\end{align*}
	for some $\bar{\beta}>0$ \cite{Limon2008,Limon2010,braendle2025}.
	Evaluating the total cost for the constructed trajectory yields
	\begin{align*}
		&J^*(\xi(t),\mathbf{\bar{v}}(t)) \leq  (1-\gamma)^2\bar{\beta} \|\xi^*_{\mathrm{a},0}(t)-\xi^\mathrm{d}(\xi^*_{\mathrm{a},0}(t)) \|^2 \\
		&+ \gamma V(M_e\xi^*_{\mathrm{a},0}(t)) + (1-\gamma)V(M_e \xi^\mathrm{d}(\xi^*_{\mathrm{a},0}(t))),
	\end{align*}
	where the last term follows from convexity of $V$.
	Since $J^*(\xi(t),\mathbf{\bar{v}}(t)) = V(M_e \xi_{\mathrm{a,0}}^*(t))$, we can sort everything with respect to $1-\gamma$ and $(1-\gamma)^2$. 
	As $V(M_e\xi_{\mathrm{a,0}}^*(t)))> V(M_e\xi^\mathrm{d}(\xi^*_{\mathrm{a},0}(t)))$, choosing $\gamma\in(\bar{\gamma},1)$ sufficiently large leads to a contraction, concluding this part of the proof.
	
	\ref{MPC:MPC:MainThmConvergence} As $p(z)\{w(t)\}=0$, we can absorb $w(t)$ in $v(t)$ and only consider the case $w(t)=0$.
	First, we show that $\lim\limits_{t\to\infty} \tau^*(t+1)-\tau^*(t)=0$ with $\tau^*(t)$ being optimal solution for $\tau$ at time $t$.
	To this end, simple calculations yield 
	\begin{align*}
		\mathbf{\hat{v}}(t)= \left(\prod\limits_{k=0}^{t-1}(1-\tau^*(k)) \right)S^t \mathbf{\hat{v}}(0)
	\end{align*}
	with $\mathbf{\hat{v}}(t) = \mathbf{\bar{v}}(t)- \mathbf{v}(t)$ and $S$ being the corresponding state transition matrix of \eqref{eq:Setup:DynamicDisturbance} such that $\mathbf{v}(t+1)=S\mathbf{v}(t)$.
	As $\tau\in[0,1]$ and $S$ is constructed from a stable system, $\mathbf{\bar{v}}(t)$ either converge to the origin, for which $\tau^*(t)=0$ is the optimal solution or $\lim\limits_{t\to\infty} \tau^*(t)=1$.
	Both imply $\lim\limits_{t\to\infty} \tau^*(t+1)-\tau^*(t)=0$.
	Now, we use the candidate solution $^\circ$ with $\tau^\circ(t+1)=\tau^*(t)$ and apply the same steps from \ref{MPC:MPC:MainThmConvergenceLam0}, such that $\lim_{t\to\infty} \xi_0^*(t) -\xi^\mathrm{d}(\xi_0^*(t))=0$. 
	Lastly, $\lim_{t\to\infty} e(t)- C_e\xi^\mathrm{d}(\xi_0^*(t))=0$ follows, as $\lim\limits_{t\to\infty} \hat{\mathbf{v}}(t+1)-S\hat{\mathbf{v}}(t)=0$. 
	Furthermore, \eqref{ass:Setup:Sylvester} ensures a unique state and input trajectory for a given $\mathbf{\hat{v}}(t)$ satisfying $\mathbf{\hat{v}}(t+1)=S\mathbf{\hat{v}}(t)$, together with the offset-free controller from Section\,\ref{sec:MPC:Controller} and $u(t) = u^*_0(t) + K(\xi(t) - \xi^*_0(t))$, we conclude $\lim_{t\to\infty} e(t)- C_e\xi^\mathrm{d}(\xi_0^*(t))=0$.
	
	\ref{MPC:MPC:NoArtRobustness} 
	Due to $w(t)\in\sigma_1\mathbb{W}$ with $\sigma_1\in[0,1)$, $p(z)$ being stable and $\mathbb{W}$ being compact, we can find an $\epsilon_1>0$ such that for all $\|\Delta v\|<\epsilon_1$ and $t\in\mathbb{N}$, it holds  $\tilde{w}(t)\coloneq w(t)+\Delta v(t)\in\mathbb{W}$ with
	\begin{align*}
		\Delta v(t)=\begin{cases}
			0, &\text{ if } t<t_0\\
			\Delta v, &\text{ if } t=t_0\\ 
			-\sum_{i=1}^{n_p}p_i\Delta v(t-i), &\text{ if } t>t_0
		\end{cases}
	\end{align*}
	such that $x(t_0+1)\in\mathbb{X}$.
	Furthermore, as $\mathbf{\bar{v}}(t_0)\in\mathbf{v}(t)\oplus \sigma_2\mathbb{V}$ with $\sigma_2\in[0,1)$, we can find an $\epsilon_2>0$ such that for all $\|\Delta v\|<\epsilon_2$, it holds $\mathbf{\bar{v}}(t_0+1)\in\mathbf{v}(t+1)\oplus \mathbb{V}$.
	Together with the artificial disturbance $\tilde{w}(t)\in\mathbb{W}$, properties \ref{MPC:MPC:MainThmConvergenceLam0} - \ref{MPC:MPC:MainThmConvergence} follow directly from their corresponding proofs and construction of the \ac{RPI}-sets.
\end{proof}
Theorem\,\ref{thm:MPC:MainResult}\, \ref{MPC:MPC:MainThmFeasibility} and \ref{MPC:MPC:MainThmConstraints} show recursive feasibility, and constraint satisfaction, when applying the proposed controller. 
Further, Theorem\,\ref{thm:MPC:MainResult} \ref{MPC:MPC:MainThmConvergenceLam0} and \ref{MPC:MPC:MainThmConvergence} provide convergence conditions to the optimal reachable output trajectory. 
Note, due to \eqref{eq:MPC:MPCInitDistState} and \eqref{eq:MPC:MPCInitContrState}, it holds $T_0(\xi(t)-\xi_0^*(t))\in\mathbb{S}$ and $T_v(\xi(t)-\xi_0^*(t))\in\mathbb{V}\oplus(-\mathbb{V})$, such that $\xi(t)$ is always within a fixed set around $\xi_0^*(t)$.
A valid initializer for $\mathbf{\bar{v}}(0)$ is $\mathbf{\bar{v}}(0)=T_v\xi(0)$.  
As before for Lemma\,\ref{thm:MPC:Model}, we require an initialization phase of $n_p$ steps to have access to $\xi(0)$. 
Hence, any change in $v_x(t)$ and $v_e(t)$ not modeled by \eqref{eq:Setup:DynamicDisturbance} requires a new initialization phase.
To address this limitation, Theorem\,\ref{thm:MPC:MainResult}\,\ref{MPC:MPC:NoArtRobustness} can be used. 
By using an enlarged set for $\mathbb{W}_x$, and $\mathbb{W}_e$, any small change in $v_x(t)$ and $v_e(t)$ can be included in $w_x(t)$ and $w_e(t)$ as long as the $\mathbf{\bar{v}}(t)$ still constitutes a feasible fallback solution.
After this step, the new disturbance state is again embedded in the extended state.
This allows the controller to adapt to slow changes in $v_x(t)$ and $v_e(t)$. 
\section{Experiment}
In this section, we validate the proposed \ac{MPC} on a four-tank system \cite{Johansson2000}. 
It  consists of two Quanser coupled tanks, as shown in Fig.\;\ref{fig:Ex:Setup}. 
\begin{figure}[t]
	\centering
	\includegraphics[width=0.32\linewidth]{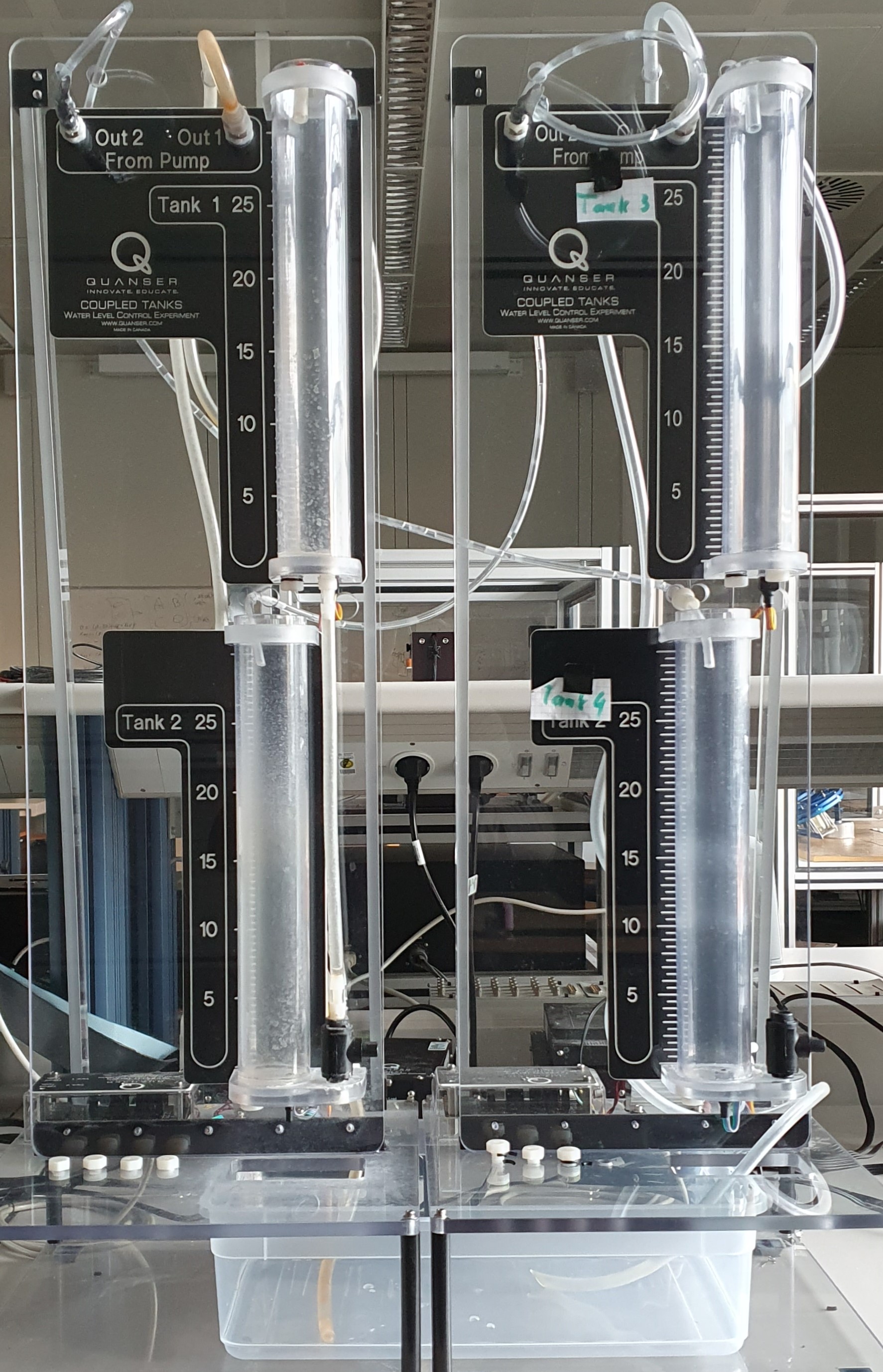}
	\caption{Four-tank system consisting of two Quanser Coupled Tanks.}
	\label{fig:Ex:Setup}\vspace{-18pt}
\end{figure}
The continuous-time linearized dynamics can be expressed by
\begin{align*}
	\begin{bmatrix}
		\dot{h}_1(t) \\ \dot{h}_2(t) \\ \dot{h}_3(t) \\ \dot{h}_4(t)
	\end{bmatrix} \!&=\!
	\begin{bmatrix}
		\!\!-a_1\!\! &    \!0\! &    \!0\! & \!0\! \\
		 \!\!a_1\!\! & \!\!-a_2\!\! &    \!0\! & \!0\! \\
		   \!0\! &    \!0\! & \!\!-a_1\!\! & \!\!0\!\!\\
		   \!0\! &    \!0\! &  \!\!a_1\!\! & \!\!-a_2\!
	\end{bmatrix} \begin{bmatrix}
	h_1(t) \\ h_2(t) \\ h_3(t) \\ h_4(t)
	\end{bmatrix} \!+\! 
	\begin{bmatrix}
		\!b_1\!\! & \!0\! \\
		\!0\!   & \!\! b_2\! \\
		\!0\!   & \!\! b_1\! \\
		\! b_2\!\! & \!0\!
	\end{bmatrix}
	\begin{bmatrix}
		u_1(t) \\ u_2(t)
	\end{bmatrix} \\
	e(t) &= \begin{bmatrix}
		h_2(t) - h_{2,\mathrm{ref}}(t)\\
		h_4(t) - h_{4,\mathrm{ref}}(t)
	\end{bmatrix}
\end{align*}
\begin{figure}[t]
	\centering
	\input{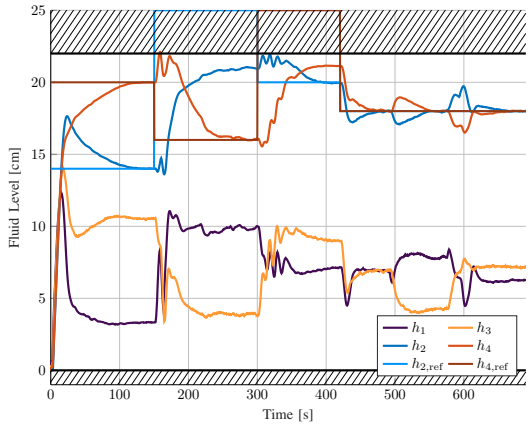}\vspace{-10pt}
	\caption{Four-tank experiment with piecewise constant references.}
	\label{fig:Ex:Real}\vspace{-18pt}
\end{figure}

with $a_1 = 0.0751\tfrac{1}{\mathrm{s}}$, $a_2=0.0371\tfrac{1}{\mathrm{s}}$, $b_1=0.151\tfrac{\mathrm{cm}}{\mathrm{V\cdot s}}$ and $b_2=\tfrac{\mathrm{cm}}{\mathrm{V\cdot s}}$. 
The operating point is $h_\mathrm{op}=[8,18,8,18]^\top\mathrm{cm}$ and $u_\mathrm{op}=[8,8]^\top\mathrm{V}$.
The states $h_1(t)$ and $h_3(t)$ describe the fluid level of the upper left and upper right tanks and $h_2(t)$ and $h_4(t)$ of the lower left and lower right tank.
We prescribe references $h_{2,\mathrm{ref}}(t)$ and $h_{4,\mathrm{ref}}(t)$ for the lower tanks. 
However, the \ac{MPC} does not require explicit knowledge of the references, it only uses the error feedback $e(t)$.
The inputs $u_1(t)$ and $u_2(t)$ are the voltage applied to two water pumps.
We discretize the system using Euler-forward discretization with a sampling time of $1\mathrm{s}$.
Furthermore, we impose the constraints $u_i(t)\in[0,16]\mathrm{V}$ for $i=1,2$ and $h_i(t)\in[0,17]\mathrm{cm}$ for $i =1,3$ and $h_i(t)\in[0,22]\mathrm{cm}$ for $i=2,4$.
We choose a prediction horizon of $N=40$ and consider constant exogenous signals, i.e.,  $p(z)=1-z^{-1}$.
We choose $Q_{\Delta x} = 0.1I$, $Q_e=5I$, $R=2I$, $P=10I$ and $\lambda = N$.
For $\mathbb{W}_x$ and $\mathbb{W}_e$, we use $w_x(t)\in[-10^{-4},10^{-4}]^4\mathrm{cm}$ and $w_e(t)\in5[-10^{-3},10^{-3}]^2\mathrm{cm}$ with controller gains
\begin{align*}
	[K_x,K_e]\!=\!\begin{bmatrix}
		\!-1.67 & \!\!\!\!-2.07 &\!\!\!\! 0.86 &\!\!\!\! 0.79 & \!\!\!\!-0.997 & \!\!\!\!0.0284 \\
		\!0.94 & \!\!\!\!0.89 & \!\!\!\!-1.76 & \!\!\!\!-2.20 & \!\!\!\!0.0332 & \!\!\!\!-0.1068
	\end{bmatrix}\!.
\end{align*}
The \ac{RPI}-sets were computed using MPT3 \cite{MPT3} with a constraint tightening of about $1\mathrm{cm}$ for the water levels and about $0.1\mathrm{V}$ for the pumps. 
As in \cite{Mayne2005}, constraint satisfaction is only guaranteed if $w(t)\in\mathbb{W}$, which is not the case for large setpoint changes.
To address this, we further impose two additional constraints in the \ac{MPC}, $u_0 + K(\xi(t)-\xi_0)\in\mathbb{U}$ and $T_v(\xi_0 - \xi(t))\in\mathbb{V}\oplus(-\mathbb{V})$ to ensure the applied input is within the constraints and to choose a sufficiently small $\tau$ to ensure that $\mathbf{\bar{v}}(t)$ stays close to $T_v\xi(t)$.
The results are illustrated in Fig.\;\ref{fig:Ex:Real}.
The \ac{MPC} successfully tracks the specified reference, and if it is not contained within the constraints, it automatically finds a close setpoint within the constraints. 
The gap between the boundary of the constraints and their respective water levels at $t=250\mathrm{s}$ and $t=400\mathrm{s}$ stems from the constraint tightening $\mathbb{X}\ominus\mathbb{S}_x$ to provide robustness against disturbances.
At $t\approx150\mathrm{s}$, a minor constraint violation occurs, as the sudden change in $e(t)$ combined with unmodeled dynamics exceeds the disturbance set.
However, the \ac{MPC} recovers.
At $t\approx 500\mathrm{s}$ and $t\approx580\mathrm{s}$, we manually opened and then closed a valve to alter the outgoing water flow of $h_1(t)$.
Despite this, the \ac{MPC} recovers the desired water level, demonstrating its ability to also reject disturbances in the dynamics.
\section{Conclusion}
In this work, we proposed a novel robust \ac{MPC} scheme to solve the output regulation problem with unknown disturbances.
By combining \ac{IMMPC} with a tube-based approach, we are able to account for disturbances generated by a known signal generator, but also for unmodeled but bounded disturbances, while ensuring constraint satisfaction, recursive feasibility and disturbance rejection.
The proposed controller has been successfully applied to a four-tank system, demonstrating its capability to reject unknown disturbances effectively.
\bibliographystyle{IEEEtran}
\bibliography{References}

\end{document}